\documentclass[12pt,a4paper]{article}
\usepackage{amssymb}
\usepackage{amscd}
\usepackage{mathrsfs}
\usepackage{eufrak}
\def\principaladviser#1{\gdef\@principaladviser{#1}}
\usepackage[centertags]{amsmath}
\usepackage{amsfonts}
\usepackage{amssymb}
\usepackage[english]{babel}
\usepackage{graphicx}
\usepackage{blindtext}
\usepackage{amsthm}
\usepackage{gnuplottex}
\usepackage{epsfig}

\usepackage{amssymb,latexsym}
\usepackage{amsfonts,amsmath}
\usepackage[T1]{fontenc}
\usepackage{amsmath,amssymb}
\usepackage{latexsym}
\usepackage{amssymb}
\usepackage{amsfonts}
\usepackage[arrow,frame,matrix]{xy}
\usepackage{fullpage}
\usepackage{hyperref}
\hypersetup{
	colorlinks=true,
	linkcolor=blue,
	filecolor=magenta,
	urlcolor=cyan,
}
\graphicspath{ {./Images/} }

\usepackage[margin=2.6cm]{geometry}
\usepackage{placeins}
\usepackage{grffile}
\usepackage{color}
\usepackage{hyperref}
\hypersetup{colorlinks,linkcolor={blue},citecolor={blue},urlcolor={red}}  
\usepackage{float}
\usepackage{subcaption}
\usepackage{ragged2e}
\usepackage[strings]{underscore}
\usepackage{parskip}

\newcommand{\Sym}{\mathrm{Sym}}

\newtheorem{theorem}{Theorem}[section]

\newtheorem{lemma}[theorem]{Lemma}
\newtheorem{corollary}[theorem]{Corollary}
\newtheorem{definition}[theorem]{Definition}
\newtheorem{proposition}[theorem]{Proposition}

\newtheorem{remark}[theorem]{Remark}

\def\be{\begin{equation}}
\def\ee{\end{equation}}
\def\bea{\begin{eqnarray}}
\def\eea{\end{eqnarray}}

\begin{document}

\title{On perfect symmetric rank-metric codes}

\author{Usman Mushrraf and Ferdinando Zullo}

\maketitle   

\begin{abstract}
Let $\Sym_q(m)$ be the space of symmetric matrices in $\mathbb{F}_q^{m\times m}$. A subspace of $\Sym_q(m)$ equipped with the rank distance is called a symmetric rank-metric code. In this paper we study the covering properties of symmetric rank-metric codes. First we characterize symmetric rank-metric codes which are perfect, i.e. that satisfy the equality in the sphere-packing like bound.
We show that, despite the rank-metric case, there are non trivial perfect codes. 
Also, we characterize families of codes which are quasi-perfect.
\end{abstract}
\section{Introduction}

Rank metric codes were first introduced by Delsarte in \cite{Delsarte:1978aa} and independently by Gabidulin in \cite{gabidulin1985theory} and Roth in \cite{roth1991maximum}. These codes have been extensively researched due to their applications in crisscross error correction \cite{roth1991maximum}, cryptography \cite{gabidulin1991ideals}, and network coding \cite{silva2008rank}; see \cite{bartz2022rank} for more applications. The coding-theoretic properties of these codes have been thoroughly studied, and optimal codes with respect to the Singleton-like bound, known as Maximum Rank Distance (MRD for short) codes, have been constructed. For more details, interested readers can refer to \cite{GorlaRavagnani,Sheekey2019}.

The study of subsets of \emph{restricted} matrices equipped with the rank metric began in 1975 by Delsarte and Goethals in \cite{delsarte1975alternating}, where they considered sets of alternating bilinear forms. The theory developed in \cite{Delsarte:1978aa} and \cite{delsarte1975alternating} also found applications in classical coding theory. Specifically, the evaluations of the forms found in \cite{delsarte1975alternating} give rise to subcodes of the second-order Reed-Muller codes, including the Kerdock codes and the chain of Delsarte–Goethals codes; see also \cite{schmidt2010symmetric}.

Using the theory of association schemes, bounds, constructions, and structural properties of restricted rank metric codes have been investigated in symmetric matrices \cite{bik2023higher,gabidulin2002representation,gabidulin2004symmetric,gabidulin2006symmetric,longobardi2020automorphism,schmidt2015symmetric,zhou2020equivalence}, alternating matrices \cite{abiad2024eigenvalue,delsarte1975alternating}, and Hermitian matrices \cite{schmidt2018hermitian,trombetti2020maximum}.

In this paper, we will focus on \textbf{symmetric rank-metric codes} in $\Sym_q(m)$, i.e. subspaces of the vector space of symmetric matrices over $\mathbb{F}_q$ of order $m$ equipped with the rank distance. The rank distance between two matrices $A$ and $B$ in $\Sym_q(m)$ is defined as 
\[ d(A,B)=\mathrm{rk}(A-B). \]
The minimum distance of a symmetric rank-metric code $C$ is defined as
\[d(C)=\min\{d(A,B) \colon A,B \in C, A\ne B\}.\]
The parameters of a symmetric rank-metric code in $\Sym_q(m)$ are $(m,|C|,d(C))$ and they are related by the following Singleton-like bound.

\begin{theorem}[\text{\cite[Theorem 3.3]{schmidt2015symmetric}}]\label{thm:SingletonBound}
Let $C$ be a symmetric rank-metric code in $\Sym_q(m)$ with minimum distance $d$ then we have
 \[\dim(C)\leq  \begin{cases}
 \frac{m(m-d+2)}{2} &\text{if } m-d \text{ is even},\\
 \frac{(m+1)(m-d+1)}{2} &\text{if } m-d \text{ is odd}.\\
  \end{cases}
  \]
\end{theorem}

The above bounds turn out to be sharp, as proved in \cite{schmidt2015symmetric}, and this allows us to give the following definition.

\begin{definition}
    Let $C$ be a symmetric rank-metric code in $\Sym_q(m)$. We say that $C$ is a \textbf{symmetric Maximum Rank Distance} (or shortly a \textbf{symmetric MRD}) code if its parameters satisfy the equality in Theorem \ref{thm:SingletonBound}.
\end{definition}

Following the classical arguments in the Hamming metric, in this paper we first prove a sphere-packing like bound and then we call perfect those codes satisfying the equality in this bound. Then we characterize the perfect symmetric rank-metric codes by proving that, despite the rank-metric (see \cite{loidreau2006properties}) and the alternating rank-metric cases (see \cite{abiad2024eigenvalue}), there exist nontrivial perfect codes. Indeed, apart from the entire space, we show that the only other family of perfect symmetric rank-metric codes is given by the symmetric MRD codes of odd order matrices and minimum distance three. Then we study the covering density of a symmetric rank-metric code $C$, which can be seen as a measure of how much the spheres centered in the codewords and of radius  $\lfloor (d(C)-1)/2\rfloor$ cover the ambient space. We give some bounds and then we define families of symmetric rank-metric codes that are quasi-perfect as those whose covering density goes to one. We characterize families of quasi-perfect codes, by proving that they exist for special sets of parameters.

\section{Bounds on sphere size}

Given a matrix $M \in \mathrm{Sym}_q(m)$, we define the \textbf{sphere} of radius $t\in \mathbb{N}_0$ centered in $M$ as
\[ S(M,t)=\{ N \in  \mathrm{Sym}_q(m) \colon \mathrm{rk}(M-N)=t\}, \]
and the \textbf{ball} of radius $t\in \mathbb{N}_0$ centered in $M$ as
\[ B(M,t)=\cup_{i=0}^t S(M,i). \]
It is easy to check that the size of a sphere and of a ball does not depend on the center and so we can denote the size of a sphere of radius $t$ by $S_t$ and the size of a ball of radius $t$ by $B_t$.
Also, note that 
\[ B_t=\sum_{i=0}^t S_i, \]
and $S_i$ represents the number of symmetric matrices of order $m$ having rank $i$ and $B_i$ represents the number of symmetric matrices of order $m$ having rank at most $i$.
The values of $S_t$ and $B_t$ depend on $q$, $m$ and $t$, more precisely the following hold (see also \cite[Theorem 2]{macwilliams1969orthogonal}).

\begin{theorem}[\text{\cite[Theorem 3]{carlitz1954sym}}] \label{lem:rkssym}
Let $0 \le t \le m$ be an integer. We have 
\[S_t=|\{M \in \Sym_m(q) \mid \mathrm{rk}(M)=t\}| = \prod_{s=1}^{\lfloor t/2 \rfloor} \frac{q^{2s}}{q^{2s}-1} \, \prod_{s=0}^{t-1}\left( q^{m-s}-1 \right).\]
\end{theorem}

As a consequence of the above theorem, we obtain that
\[ B_t=\sum_{i=0}^t \left(\prod_{s=1}^{\lfloor i/2 \rfloor} \frac{q^{2s}}{q^{2s}-1} \, \prod_{s=0}^{i-1}\left( q^{m-s}-1 \right)\right). \]

We will now provide some upper and lower bounds on the size of a sphere and of a ball. 

\begin{proposition}\label{prop:boundvolume}
For any $i \in \{0,\ldots,m\}$ and for any $q$, we have
\[q^{(m-1)i-\frac{i(i-1)}{2}}< S_i\leq q^{mi-\frac{i(i-1)}{2}+\lfloor \frac{i}{2}\rfloor},\]
and
\[q^{(m-1)i-\frac{i(i-1)}{2}}\leq B_i\leq q^{mi-\frac{i(i-1)}{2}+\lfloor \frac{i}{2}\rfloor+1}.\]
\end{proposition}
\begin{proof}
The assertion is clearly true when $i=0$. So, assume that $i\geq 1$ and let us start by determining the bounds on 
 \[S_{i}=\prod^{\lfloor i/2 \rfloor}_{s=1}\frac{q^{2s}}{q^{2s}-1}\prod^{i-1}_{s=0}(q^{m-s}-1).\]
Since     
\[\prod^{i-1}_{s=0}(q^{m-s}-1)\leq\prod^{i-1}_{s=0}q^{m-s}= q^{mi-\frac{i(i-1)}{2}}\]
and \[\frac{q^{2s}}{q^{2s}-1}\leq q,\]
it is easy to see that
\[\prod^{\lfloor \frac{i}{2} \rfloor}_{s=1}\frac{q^{2s}}{q^{2s}-1}\leq \prod^{\lfloor \frac{i}{2} \rfloor}_{s=1}q=q^{\lfloor\frac{i}{2}\rfloor}.\]
So, 
\begin{equation}\label{eq:1}
S_{i}\leq q^{\lfloor \frac{i}{2}\rfloor}\cdot q^{mi-\frac{i(i-1)}{2}}=q^{
mi-\frac{i(i-1)}{2}+\lfloor \frac{i}{2} \rfloor}.
\end{equation}
For lower bound on $S_i$, observe that
\[q^{m-s}-1\geq q^{m-s-1},\]
from which
 \begin{equation*}
  \prod^{i-1}_{s=0}(q^{m-s}-1) \geq \prod^{i-1}_{s=0}q^{m-s-1}=q^{(m-1)i-\frac{i(i-1)}{2}}.
 \end{equation*}
 Using that,
 \[\frac{q^{2s}}{q^{2s}-1}> 1 \,\,\,\text{and}\,\,\, \prod^{\lfloor i/2\rfloor}_{s=1}\frac{q^{2s}}{q^{2s}-1}> 1,\]
 we get
 \begin{equation}\label{eq:2}
 S_{i}=\prod^{\lfloor \frac{i}{2}\rfloor}_{s=1}\frac{q^{2s}}{q^{2s}-1}\prod^{i-1}_{s=0}(q^{m-s}-1)> 1\cdot q^{(m-1)i-\frac{i(i-1)}{2}}=q^{
(m-1)i-\frac{i(i-1)}{2}}
\end{equation}
 From \eqref{eq:1} and \eqref{eq:2} we obtain the desired bounds
 \[q^{(m-1)i-\frac{i(i-1)}{2}}< S_i\leq q^{mi-\frac{i(i-1)}{2}+\lfloor \frac{i}{2}\rfloor}.\]
Let us know analyze the size of the balls. Since $B_i=\sum_{s=0}^i S_s \geq S_i$, the lower bound on $B_i$ comes directly from the lower bound on the size of sphere of radius $i$.
So, we only need to prove the upper bound. Because of the upper bound on $S_i$, we have
\begin{align*}
B_{i}&\leq \sum^{i}_{s=0}q^{ms-\frac{s(s-1)}{2}+\lfloor\frac{s}{2}\rfloor}&\\
&=q^{mi-\frac{i(i-1)}{2}+\lfloor\frac{i}{2}\rfloor}\left(1+\sum^{i-1}_{s=0}\frac{q^{ms-\frac{s(s-1)}{2}+\lfloor\frac{s}{2}\rfloor}}{q^{mi-\frac{i(i-1)}{2}+\lfloor\frac{i}{2}\rfloor}}\right)&\\
&=q^{mi-\frac{i(i-1)}{2}+\lfloor\frac{i}{2}\rfloor}\left(1+\sum^{i-1}_{s=0}q^{ms-mi+\frac{i^2-s^2}{2}-\frac{i-s}{2}+\lfloor\frac{s}{2}\rfloor-\lfloor\frac{i}{2}\rfloor}\right).
\end{align*}
As, for every $s\leq i-1$,
\begin{equation*}
    ms-mi+\frac{i^2-s^2}{2}-\frac{i-s}{2}+\left\lfloor\frac{s}{2}\right\rfloor-\left\lfloor\frac{i}{2}\right\rfloor
    \leq 
    (i-s)\left(i-m\right),
\end{equation*}
since
\[ \left\lfloor\frac{s}{2}\right\rfloor\leq \left\lfloor \frac{i}{2}\right\rfloor\,\,\, \text{and}\,\,\, s\leq i,\]
we have that
\begin{equation}\label{eq:boundBi}
B_{i}\leq q^{mi-\frac{i(i-1)}{2}+\lfloor\frac{i}{2}\rfloor}\left(1+\sum^{i-1}_{s=0}q^{(i-s)(i-m)}\right)\leq q^{mi-\frac{i(i-1)}{2}+\lfloor\frac{i}{2}\rfloor}\left(1+\sum^{i}_{j=1}q^{j(i-m)}\right).
\end{equation}
Observe that
\[1+\sum^{i}_{j=1}q^{j(i-m)}=\frac{1-q^{(i+1)(i-m)}}{1-q^{i-m}}\leq \frac{1}{1-q^{i-m}}\leq q,\]
so that, together with \eqref{eq:boundBi}, we have
\[B_{i}\leq q^{mi-\frac{i(i-1)}{2}+\lfloor\frac{i}{2}\rfloor}\cdot q,\]
which concludes the proof.
\end{proof}

\begin{remark}
    The above bounds can be certainly improved. Because of our purposes of the next section, we need to upper bound $B_i$ by a power of $q$.
\end{remark}

In the next section we will use these bounds to characterize perfect codes in the theory of symmetric rank-metric codes.

\section{Perfect symmetric rank-metric codes}

In this section we will explore the sphere-packing bound and we will investigate the parameters for which perfect codes exist in the symmetric matrices framework. 

\begin{theorem}[Sphere-packing bound]
Let $C\subseteq \Sym_{q}(m)$ be a symmetric rank metric code of minimum distance $d$ and $|C|=M$. If $t=\lfloor\frac{d-1}{2}\rfloor$, then 
\[MB_{t}\leq q^{\frac{m^2+m}{2}}.\]
\end{theorem}
\begin{proof}
    Let $c_{1},c_{2},...,c_{M}$ be the codewords of $C$. Using the triangle inequality of the rank metric, we can see that 
    \[B_{t}(c_{i})\cap B_{t}(c_{j})=\varnothing,\] for any $i,j \in \{1,\ldots,M\}$ with $i\neq j$.
    So,
    \[\left|\bigcup^{M}_{i=1}B_{t}(c_{i})\right|=\bigcup^{M}_{i=1}\left| B_{t}(c_{i})\right|\leq |\Sym_{q}(m)|=q^{\frac{m^2+m}{2}}\]
    implying that
    \[\sum^{M}_{i=1}B_{t}\leq q^{\frac{m^2+m}{2}},\]
    and hence the desired bound. 
\end{proof}

As for the already studied metrics, we give the following definition of perfect code.

\begin{definition}
    Let $C$ be a symmetric rank-metric code in $\Sym_q(m)$. We say that $C$ is a \textbf{perfect} if its parameters satisfy the equality in the sphere-packing bound.
\end{definition}

A trivial example is obtained when considering $C=\Sym_q(m)$, indeed in this case its minimum distance is $1$ and so $t=(d-1)/2=0$ and $B_0=1$.
We will now show, through some steps, that a code in $\Sym_q(m)$ with minimum distance $d$ is perfect if and only if either it is the trivial code (i.e. $\Sym_q(m)$) or it is a symmetric MRD code, with $m$ odd and $d=3$.
We first analyze the case in which $\lfloor\frac{d-1}{2}\rfloor=1$.

\begin{lemma}\label{lem:t=1}
    Let $C\subseteq \Sym_{q}(m)$  be a perfect code with minimum distance $d$ and let $t=\lfloor \frac{d-1}{2}\rfloor=1$. Then  $d=3$, $m$ is odd and $C$ is a symmetric MRD code.
\end{lemma}
\begin{proof}
Note that in this case, $d \in \{3,4\}$.
Since $t=1$ and $B_{1}=q^{m}$, we have that
\[M=q^{\frac{m^2-m}{2}}.\]
When $m-d$ is even, Theorem \ref{thm:SingletonBound} implies
\[q^{\frac{m^2-m}{2}}=M\leq q^{\frac{m(m-d+2)}{2}},\]
that is
\[d\leq 3,\]
and so $d=3$ and $m$ is odd. When $d=3$ and $m$ is odd, we have the equality in the sphere-packing bound if and only if $C$ is a symmetric MRD code. 
When $m-d$ is odd, Theorem \ref{thm:SingletonBound} implies
 \[q^{\frac{m^2-m}{2}}=M\leq q^{\frac{(m+1)(m-d+1)}{2}},\]
that is
\begin{equation*}
  d(m+1)\leq 3m+1,   
\end{equation*}
a contradiction to $d \in \{3,4\}$.
\end{proof}

In the next lemma we will analyze the case in which $\lfloor \frac{d-1}{2}\rfloor=2$.

\begin{lemma}\label{lem:t=2}
    Let $C\subseteq \Sym_{q}(m)$  be a symmetric rank-metric code with minimum distance $d$ and let $t=\lfloor\frac{d-1}{2}\rfloor=2$ then $C$ is not perfect.
\end{lemma}
\begin{proof}
Suppose on contrary that $C$ is perfect, then
\[MB_{2}=q^{\frac{m^2+m}{2}}.\]
Note that $M=q^i$ for some natural number $i<(m^2+m)/2$, as $t=2$ and $C$ cannot be equal to $\Sym_{q}(m)$. Therefore, we get that
\begin{equation}\label{eq:B2i}
B_{2}= q^{\frac{m^2+m}{2}-i}.
\end{equation}
Since $B_{2}=q^{m}+\frac{q^2(q^m-1)(q^{m-1}-1)}{q^2-1}$ then if $m>2$ we have $B_2/q^2\equiv -1 \pmod{q}$, so we get a contradiction from \eqref{eq:B2i}.
When $m=2$ we still get a contradiction from \eqref{eq:B2i} and the size of $B_2$.
\end{proof}

We can now completely determine the parameters of a perfect code in $\Sym_q(m)$.

\begin{theorem}\label{thm:2.4}
    Let $C\subseteq \Sym_{q}(m)$. Then $C$ is a perfect code if and only if $C=\Sym_{q}(m)$ or $d(C)=3$, $m$ is odd and $C$ is a symmetric MRD code.
\end{theorem}
    \begin{proof}
        Suppose that $m-d$ is even,  then 
        from Theorem \ref{thm:SingletonBound} we have 
        \[M\leq q^{\frac{m(m-d+2)}{2}}, \]
        and using the fact that $C$ is perfect we have 
        \[q^{\frac{m(m-d+2)}{2}}B_{t}\geq q^{\frac{m^2+m}{2}}, \]
        with $t=\lfloor (d(C)-1)/2\rfloor$.
        Since, by Proposition \ref{prop:boundvolume} $B_{t}\leq q^{mt-\frac{t(t-1)}{2}+\frac{t}{2}+1}$, then 
        
        \[q^{\frac{m^2+m}{2}}\leq q^{mt-\frac{t(t-1)}{2}+\frac{t}{2}+1+\frac{m(m-d+2)}{2}},\]
        that is
        
        
         \begin{equation}\label{eq:condm2t+1noperf}
        \frac{m^2+m}{2}\leq mt-\frac{t(t-1)}{2}+\frac{t}{2}+1+\frac{m(m-d+2)}{2}.
        \end{equation}
        Assume that $d=2t+1$, then \eqref{eq:condm2t+1noperf} reads as follows
        \[0\leq -t^2+2t+2,\]
        which holds only for $t\in\{0,1,2\}$.\\
Now, consider $d=2t+2$, then \eqref{eq:condm2t+1noperf} implies that
\[m\leq -t^2+2t+2.\]
Since $m\geq 2$, again we have that $t\in \{0,1,2\}$.
Arguing as before in the case in which $m-d$ is odd, we have again that $t\in \{0,1,2\}$.
If $t=0$, then $d=1$ and $C$ is perfect if and only if $C= \Sym_{q}(m)$.
The remaining part follows by Lemmas \ref{lem:t=1} and \ref{lem:t=2}.
\end{proof}

\subsection{Examples of non-trivial perfect codes}

In this section we will recall the known examples of symmetric MRD codes. Their description is given in terms of linearized polynomials, which is an equivalent framework as we will see in the next lines. 
Let $\mathcal{L}_{m, q}$ denote the quotient $\mathbb{F}_q$-algebra of all \textbf{$q$-polynomials} (or \textbf{linearized polynomials}) over $\mathbb{F}_{q^m}$ with degree smaller than $q^m$, namely,
\[\mathcal{L}_{m,q}=\left\{ \sum\limits_{i=0}^{m-1} a_i x^{q^i} \colon a_i \in \mathbb{F}_{q^m} \right\}. \]
It is well known that the $\mathbb{F}_{q}$-algebra $\mathcal{L}_{m, q}$ is isomorphic to the $\mathbb{F}_{q}$-algebra $\mathbb{F}_q^{m\times m}$; see e.g. \cite{WU201379}.
In \cite{longobardi2020automorphism} symmetric matrices are described in terms of linearized polynomials. 
Indeed, the vector space $\Sym_q(m)$ can be identified to 
\[{\mathcal{S}}_q(m)=\left\{ \sum\limits_{i=0}^{m-1} c_i x^{q^{i}} \colon c_{m-i}=c_i^{q^{m-i}} \,\text{ for } \, i \in \{0,\ldots,m-1\} \right\}\subseteq \mathcal{L}_{m,q}. \]

In \cite[Theorem 4.4]{schmidt2015symmetric} the author proved that for all the admissible values of the minimum distance, $q$ and $m$ there are symmetric MRD codes (see also \cite{longobardi2020automorphism} for the description in terms of linearized polynomials).

\begin{theorem}[\text{\cite[Theorem 4.4]{schmidt2015symmetric}}]\label{Th-Schmidt-4.4}
	Let $m$ and $d$ be two positive integers such that $1\leq d\leq m$ and $m-d$ is even. 
	The set
	\begin{equation}\label{ex:E1s}
	\mathcal{S}_{q,m,d}=\left\{ b_0x + \sum\limits_{j=1}^{\frac{m-d}{2}} \left( b_jx^{q^j}+(b_jx)^{q^{m-j}}  \right) \colon b_0,\ldots,b_{\frac{m-d}{2}} \in \mathbb{F}_{q^m}\right\}
	\end{equation}
	defines a symmetric MRD code in $\mathrm{Sym}_q(m)$. 
\end{theorem}

In \cite[Theorem 4.1]{schmidt2015symmetric} it has been shown that constructions of symmetric MRD codes with minimum distance $d$ in $\Sym_q(m)$ with $m-d$ odd can be obtained by puncturing (in terms of matrices, removing rows and columns) the above examples $\mathcal{S}_{q,m,d}$ of symmetric MRD codes.
When $d=3$ and $m$ is odd, $\mathcal{S}_{q,m,d}$ is a perfect code.
Moreover, for this families of codes fast encoding and decoding algorithms have been developed; see \cite{gabidulin2002representation}, \cite{gabidulin2003transposed}, \cite{gabidulin2004symmetric}, \cite{gabidulin2006symmetric} and \cite{kadir2022encoding}.
In \cite[Section 5]{longobardi2020automorphism}, also another construction of symmetric MRD codes has been found, but in this case the codes are not perfect.

\section{Covering density}

In the previous sections, we showed that existence of perfect codes in symmetric rank-metric context depends strongly on the order of matrices, the dimension of the code and minimum distance of code. 
However, a natural investigation regards the study of covering density of a symmetric rank-metric code $C$, which gives a measure of how much of the ambient space is covered by the spheres centered in the codewords and of radius $\lfloor \frac{d(C)-1}{2}\rfloor$.

\begin{definition}
    Let $C$ be a symmetric rank-metric code in $\Sym_q(m)$ of size $M$. The \textbf{covering density} of $C$ is defined as 
\[D(C)=\frac{MB_{t}}{q^{\frac{m^2+m}{2}}},\] where $t=\lfloor \frac{d(C)-1}{2}\rfloor$.
\end{definition}

As a consequence of the bounds in Proposition \ref{prop:boundvolume}, we have the following bound for the covering density of a symmetric MRD code.

\begin{proposition}\label{prop:3.1}
Let $C$ be symmetric MRD code in $\Sym_q(m)$ with minimum distance $d$ and let $t=\lfloor \frac{d(C)-1}{2}\rfloor$. 
The covering density $D$ of $C$ is upper bounded as follows 
\[D\leq  \begin{cases}
 q^{-t-\frac{t(t-1)}{2}+\lfloor \frac{t}{2}\rfloor+1} &\text{if } m \text{ is even and } d=2t+1, \\ 
  q^{-t-\frac{t(t-1)}{2}+\lfloor \frac{t}{2}\rfloor-\frac{m}{2}+\frac{1}{2}} &\text{if } m \text{ is odd and } d=2t+2, \\
   q^{-\frac{t(t-1)}{2}+\lfloor \frac{t}{2}\rfloor+1} &\text{if } m \text{ is odd and } d=2t+1, \\
    q^{-\frac{t(t-1)}{2}+\lfloor \frac{t}{2}\rfloor-\frac{m}{2}+1} &\text{if } m \text{ is even and } d=2t+2, 
  \end{cases}
  \]
and lower bounded as follows
\[D\geq  \begin{cases}
 q^{-2t-\frac{t(t-1)}{2}} &\text{if } m \text{ is even and } d=2t+1, \\
 q^{-2t-\frac{t(t-1)}{2}-\frac{m}{2}-\frac{1}{2}} &\text{if } m \text{ is odd and } d=2t+2, \\
   q^{-t-\frac{t(t-1)}{2}} &\text{if } m \text{ is odd and } d=2t+1, \\
    q^{-t-\frac{t(t-1)}{2}-\frac{m}{2}} &\text{if } m \text{ is even and } d=2t+2. 
  \end{cases}
  \]
\end{proposition}

It is worth to mention that when the minimum distance is odd, the upper and lower bounds do not depend on the order of the matrices involved but only on the minimum distance.
Clearly, using Theorem \ref{thm:SingletonBound}, the upper bounds in Proposition \ref{prop:3.1} immediately reads as follows for a symmetric rank-metric code.

\begin{corollary}\label{Cor:1}
Let $C$ be a symmetric rank-metric code in $\Sym_q(m)$ with minimum distance $d$ and let $t=\lfloor \frac{d(C)-1}{2}\rfloor$. The covering density \[D\leq  \begin{cases}
 q^{-t-\frac{t(t-1)}{2}+\lfloor \frac{t}{2}\rfloor+1} &\text{if } m \text{ is even and } d=2t+1, \\
  q^{-t-\frac{t(t-1)}{2}+\lfloor \frac{t}{2}\rfloor-\frac{m}{2}+\frac{1}{2}} &\text{if } m \text{ is odd and } d=2t+2, \\
   q^{-\frac{t(t-1)}{2}+\lfloor \frac{t}{2}\rfloor+1} &\text{if } m \text{ is odd and } d=2t+1, \\
    q^{-\frac{t(t-1)}{2}+\lfloor \frac{t}{2}\rfloor-\frac{m}{2}+1} &\text{if } m \text{ is even and } d=2t+2. 
  \end{cases}
  \]
\end{corollary}

We did not mention the lower bounds as they strongly depends on the size of the codes considered.
As the exponents appearing in the bounds of Proposition \ref{prop:3.1} are relatively small, these bounds also suggest that there should not exist quasi-perfect codes, apart from the cases in which we know that there are perfect codes.
As classically done, see e.g. \cite{etzion2005quasi} and \cite{loidreau2006properties} (for the rank-metric version), we give the definition of quasi-perfect families of codes. 

\begin{definition} 
    Let $\mathcal{F}=\{C_i\}_i$ be a family of $(m_i,|C_{i}|,d)$ symmetric rank-metric codes and let $D_{i}$ be the covering density of code $C_{i}$. If
    \[\lim_{i\to\infty}D_{i}=1,\]
    then we say that the family $\mathcal F$ is \textbf{quasi-perfect}.
\end{definition}

We will now show that, apart from the parameters from which we know there exist perfect codes, there are no quasi-perfect families of codes.

\begin{theorem}
    There are no quasi-perfect families of codes with minimum distance $d$ such that 
    \[ \left\lfloor\frac{d-1}2 \right\rfloor\geq 3 .\]
\end{theorem}
\begin{proof}
Let $\mathcal{F}=\{C_i\}_i$ be a family of $(m_i,|C_{i}|,d)$ symmetric rank-metric codes and let 
\[D_{i}=\frac{|C_{i}||B_{t}|}{q^{\frac{i^2+i}{2}}}\]
be the covering density of code $C_{i}$, where $t=\lfloor (d(C_i)-1)/2\rfloor$. 
We first show that $D_{i}\leq 1/q=q^{-1}$ for any $i$. 
From Corollary \ref{Cor:1}, when $d=2t+1$ and $i$ is even we have
\[D_{i}\leq q^{-t-\frac{t(t-1)}{2}+\lfloor \frac{t}{2}\rfloor+1}\leq q^{-1},\]
if and only if $t\geq 2$.
Arguing in a similar way for the remaining cases, we have that  $D_{i}\leq q^{-1}$ for any $i$. 

Therefore if $\lim_{i\to \infty}D_i$ exists, it is at most $1/q$, and so it cannot be one.
\end{proof}

We will now analyze the cases in which $t=\left\lfloor\frac{d-1}2 \right\rfloor \in \{0,1,2\}$.
When $t=0$ then we have either $d=1$ or $d=2$. In this case we can consider the family of codes $\{\Sym_q(i)\}_i$, and since the covering densities of the considered codes are $1$, this family defines a (trivial) quasi-perfect family of codes.
When $d=2$ the covering density is given by the ratio of the size code and size of the ambient space, and it is easy to see that in this case we cannot have quasi-perfect families of codes.

When $t \in \{1,2\}$ we need a deeper analysis of the covering density.
Indeed, using Theorem \ref{lem:rkssym} and the fact that
\[ B_2=\frac{-q^m+q^{2m+1}-q^{m+1}+q^2}{q^2-1}, \]
we obtain the following.

\begin{proposition}\label{prop:3.2}
Let $C$ be a symmetric rank-metric code in $\Sym_q(m)$ with minimum distance $d$ and let $t=\left\lfloor\frac{d-1}2 \right\rfloor$. The covering density $D$ of $C$ is upper bounded as follows for $t=1$
\[D\leq  \begin{cases}
 1 &\text{if } m \text{ is odd and } d=3, \\
  q^{-1} &\text{if } m \text{ is even and } d=3, \\
   q^{-\frac{m}{2}} &\text{if } m \text{ is even and } d=4, \\
    q^{-\frac{m+3}{2}} &\text{if } m \text{ is odd and } d=4, 
  \end{cases}
  \]
and for $t=2$
  \[D\leq  \begin{cases}
 \frac{-q^{-m}+q-q^{-m+1}+q^{2-2m}}{q^2-1} &\text{if } m \text{ is odd and } d=5, \\
  \frac{-q^{-m-2}+q^{-1}-q^{-m-1}+q^{-2m}}{q^2-1}  &\text{if } m \text{ is even and } d=5, \\
   \frac{-q^{\frac{-3m}{2}}+q^{\frac{-m+2}{2}}-q^{\frac{-3m+2}{2}}-q^{\frac{-5m+4}{2}}}{q^2-1} &\text{if } m \text{ is even and } d=6, \\
    \frac{-q^{\frac{-3m-5}{2}}+q^{\frac{-m-3}{2}}-q^{\frac{-3m-3}{2}}-q^{\frac{-5m-1}{2}}}{q^2-1} &\text{if } m \text{ is odd and } d=6. 
  \end{cases}
  \]
If $C$ is a symmetric MRD code, then the equality holds.
\end{proposition}

Using the above proposition we can exclude the existence of family of quasi-perfect codes having minimum distance either $5$ or $6$.

\begin{proposition}
There are no quasi-perfect families of symmetric rank-metric codes with minimum distance $d$ such that 
\[ \left\lfloor\frac{d-1}2 \right\rfloor=2 .\]
\end{proposition}
\begin{proof}
    Let $\{C_{i}\}_i$ be a family of symmetric rank-metric codes with minimum distance $d$ where $C_i$ has covering density $D_{i}$. 
    A case by case analysis shows that, using the upper bounds in Proposition \ref{prop:3.2}, $D_i$ is strictly less than $1/q$. So, when considering the limit of the $D_i$'s, if it exists, it is bounded by $1/q$ and it cannot be one.
    Therefore it can be concluded that there are no quasi-perfect families of codes with minimum distance $d$ such that $\lfloor (d-1)/2\rfloor=2$.
\end{proof}

When $t=1$, quasi-perfect families of codes exist. Indeed, if $C_i$ is a symmetric MRD code with parameters $(2i+1,q^{\frac{i(i-1)}{2}},3)$, then $D(C_i)=1$ for any $i$ by Theorem \ref{thm:2.4} and so 
\[ \lim_{i\to\infty} D_i=1. \]

\begin{proposition}
    Let $\{C_i\}_i$ be a family of symmetric rank-metric codes where $C_i$ has parameters $(m_i,|C_i|,d)$ and $d \in \{3,4\}$.
    \begin{itemize}
        \item  If $d=3$ and the sequence $\{m_i\}_i$ has at least one subsequence of even numbers, then $\{C_i\}_i$ is not quasi-perfect.
        \item If $d=4$ then $\{C_i\}_i$ is not quasi-perfect.
    \end{itemize}
\end{proposition}
\begin{proof}
    Suppose that $d=3$.
    Let $\{m_j\}_{j\in \mathcal{J}}$ be a subsequence of $\{m_i\}_i$ with the property that $m_j$ is an even number for any $j \in \mathcal{J}$.
    From Proposition \ref{prop:3.2}, when $d=3$ and $m_j$ is even the covering density is bounded by $1/q$ and therefore
    if $\lim_{j \to \infty} D(C_j)$ exists then it is bounded by $1/q$ and so either $\lim_{i\to \infty} D(C_i)$ does not exist or it is at most $1/q$. 
    Assume now that $d=4$, then by Proposition \ref{prop:3.2} for any $i$ we have that
    \[ D(C_i)\leq q^{-\frac{m_i}2}, \]
    and therefore $\lim_{i\to \infty} D(C_i)=0$.
\end{proof}

Combining the above results we have the following classification of quasi-perfect families of codes.

\begin{corollary}
    There are quasi-perfect families of symmetric rank-metric codes with minimum distance $d$ if and only if $d\in \{1,3\}$.
    The latter case only happens when for the families of codes that does not contain infinitely many code of even order.
\end{corollary}

\section*{Acknowledgements}

The research of the last two authors was partially supported by the project COMBINE of ``VALERE: VAnviteLli pEr la RicErca" of the University of Campania ``Luigi Vanvitelli'' and was partially supported by the INdAM - GNSAGA Project \emph{Tensors over finite fields and their applications}, number E53C23001670001.

\bibliographystyle{abbrv}
\bibliography{biblio}

Usman Mushrraf and Ferdinando Zullo\\
Dipartimento di Matematica e Fisica,\\
Universit\`a degli Studi della Campania ``Luigi Vanvitelli'',\\
I--\,81100 Caserta, Italy\\
{{\em \{usman.mushrraf,ferdinando.zullo\}@unicampania.it}}
\end{document}